\documentclass{article}
\usepackage[utf8]{inputenc}
\usepackage{amsmath}
\usepackage{amsfonts}
\usepackage{amssymb}
\usepackage{graphicx}
\usepackage[title,titletoc]{appendix}
\usepackage{fullpage}
\usepackage{color}
\usepackage[sort&compress,numbers]{natbib}
\usepackage{caption}
\usepackage{subcaption}
\usepackage{float}
\usepackage[normalem]{ulem} 
\usepackage{multirow}
\usepackage{caption}
\usepackage{bm}
\usepackage{array}

\usepackage{amsthm}
\newtheorem{theorem}{Theorem}[section]

\newtheorem{prop}{Proposition}

\numberwithin{subcase}{case}

\newtheorem{remark}[theorem]{Remark}

\usepackage{algorithm} 
\usepackage{algpseudocode} 

\usepackage{xcolor} 

\newcommand{\keywords}[1]{%
  \vspace{0.5cm}
  \noindent \textbf{Keywords:} #1
}

\begin{document}
\algnewcommand{\algorithmicgoto}{\textbf{go to}}%
\algnewcommand{\Goto}[1]{\algorithmicgoto~\ref{#1}}%
\title{\Large\bf \vspace{-8ex}
Spline tie-decay temporal networks}
\author{Chanon Thongprayoon$^1$, Naoki Masuda$^{1,2,3,*}$}
\date{\vspace{-5ex}}
\maketitle
\vspace{-14pt}
\begin{center}
\small
$^1$Department of Mathematics, State University of New York at Buffalo, Buffalo, NY, USA.\\
$^2$Institute for Artificial Intelligence and Data Science, State University of New York at Buffalo, Buffalo, NY, USA\\
$^3$Center for Computational Social Science, Kobe University, Kobe, Japan\\
*Corresponding author: naokimas@buffalo.edu
\end{center}

\begin{abstract}
Increasing amounts of data are available on temporal, or time-varying, networks. There have been various representations of temporal network data each of which has different advantages for downstream tasks such as mathematical analysis, visualizations, agent-based and other dynamical simulations on the temporal network, and discovery of useful structure. The tie-decay network is a representation of temporal networks whose advantages include the capability of generating continuous-time networks from discrete time-stamped contact event data with mathematical tractability and a low computational cost. However, the current framework of tie-decay networks is limited in terms of how each discrete contact event can affect the time-dependent tie strength (which we call the kernel). Here we extend the tie-decay network model in terms of the kernel. Specifically, we use a cubic spline function for modeling short-term behavior of the kernel and an exponential decay function for long-term behavior, and graft them together. This spline version of tie-decay network enables delayed and $C^1$-continuous interaction rates between two nodes while it only marginally increases the computational and memory burden relative to the conventional tie-decay network. We show mathematical properties of the spline tie-decay network and numerically showcase it with three tasks: network embedding, a deterministic opinion dynamics model, and a stochastic epidemic spreading model.
\end{abstract}


\section{Introduction\label{sec:introduction}}

Various complex systems are modeled by networks consisting of nodes and edges~\cite{barabasi2016network, newman2018networks}. In fact, many networks evolve over time much like the case of dynamical processes occurring on them. Such networks are known as temporal networks~\cite{bansal2010thedynamic, holme2012temporal, holme2015modern, masuda2020guide}. Examples of temporal networks include contact networks among human or animal individuals whose dynamics of connectivity are in most cases driven by mobility of individuals, and air transportation networks in which time-stamped flights connect pairs of airports (i.e., nodes) in a time-dependent manner as the flight schedule changes over weeks and months. Incorporating temporality into network analysis has proved to be successful in performing downstream tasks. For instance, time-dependent communities in temporal networks can be more informative than communities in static networks, with the former providing features such as birth, death, merger, and split of communities~\cite{palla2007quantifying, fortunato2010community}. Other examples of the utility of temporal network analysis are link prediction~\cite{lu2011link, tabourier2016predicting} (but see~\cite{he2024sequential}) and network control~\cite{li2017fundamental}.

There are various representations of temporal networks, some of which are lossy and others are not \cite{holme2015modern, masuda2020guide, gauvin2022randomized}. In the simplest case of the so-called event-based representation of temporal network, we ignore the duration of each time-stamped event and view a temporal network as a set of triplets. A triplet encodes a time-stamped edge, or contact event by specifying the two nodes involved in the event and the time of the event. Empirical data are given in this format in a majority of cases. Another popular representation of temporal networks is the snapshot representation, with which one considers a temporal network as a sequence of static networks switching from one to another at discrete points of time. Between two consecutive switches, the network stays static. Mathematically, a sequence of static networks can be represented by a sequence of adjacency, Laplacian, or other matrices, which often allows us to exploit matrix algebra techniques to theoretically analyze temporal networks. The snapshot representation is essentially a discrete-time representation in the sense that the network is piecewise constant over time. The tie-decay network model enables one to construct a genuinely continuous-time time-dependent adjacency matrix given a sequence of time-stamped events as input.
Tie-decay networks assume that the edge weight instantaneously increases by $1$ upon an event arrival and that the edge weight exponentially decays over time in the absence of events~\cite{ahmad2021tie}.
This assumption is reminiscent of the Hawkes process,
which is a self-exciting point process model~\cite{vere1970stochastic, hawkes1971point, hawkes1971spectra, masuda2013self-exciting, laub2015hawkes}, with the time-dependent edge weight corresponding to the event rate of the non-homogeneous Poisson process generating event sequences.
The tie-decay networks have been used for the PageRank centrality~\cite{porter2020nonlinearity+, ahmad2021tie}, opinion dynamics modeling~\cite{sugishita2021opinion}, diffusion dynamics~\cite{zuo2021models}, epidemic dynamics~\cite{chen2022epidemic}, and network embedding~\cite{thongprayoon2023embedding, thongprayoon2023online}.

There are limitations of the current tie-decay network model. First, the tie-decay network model gives rise to discontinuity of edge weights upon event arrivals. In contrast, some temporal network analyses assume the continuity of edge weights in time such as adaptive network dynamics~\cite{sayama2013modeling, berner2023adaptive}. Second, there are situations in which the effect of discrete events does not instantly come in action such as epidemiological dynamics~\cite{ma2004global, kumar2020deterministic, liu2020covid}, social dynamics~\cite{olfatisaber2004consensus, atay2013consensus}, and neuronal dynamics~\cite{rall1967distinguishing, destexhe1994synthesis, van1994inhibition}.

To alleviate these limitations, we extend the tie-decay network model by allowing the edge weight (i.e., tie strength) to increase with a delay in response to a contact event on the edge and to avoid instantaneous jumps in the edge weight. Specifically, we continuously and gradually increase the edge weight, which is assumed to exponentially decay at large delay. To realize this property, we graft a cubic spline polynomial and an exponential tail to represent the effect of a single contact event on the edge weight over time. This combination allows us to maintain a concise and memory- and time-efficient updating rule for the edge weight owing to the exponential tail part, while resolving the two aforementioned problems thanks to the spline part. We also provide some mathematical properties of the proposed tie-decay network model and numerically demonstrate the method with network embedding, an opinion spreading dynamics model, and an epidemic dynamics model.

\section{Tie-decay network models}

In this section, we introduce the tie-decay network and propose its extension using cubic spline polynomials.

\subsection{Exponential tie-decay network}\label{sec:tie-decay}

We first explain the original tie-decay network model, in which the edge weight exponentially decays over time in the absence of external input~\cite{ahmad2021tie}. Consider an edge $(i, j)$ between the $i$th and $j$th nodes in the given network with $N$ nodes.
We assume that input data are a sequence of time-stamped contact events on any edge. We represent such a sequence on edge $(i, j)$ by $\{ \tilde{t}_1, \tilde{t}_2, \ldots \}$, where $0\le \tilde{t}_1 \le \tilde{t}_2 \le \cdots$, and $\tilde{t}_{\ell}$ is the time of the $\ell$th event on edge $(i, j)$; we do not indicate the edge index in $\tilde{t}_{\ell}$ for simplifying the notation.
We denote the $N\times N$ time-dependent weighted adjacency matrix of the tie-decay network by $B(t)=[b_{ij}(t)]$, where $b_{ij}(t)$ represents the weight of the edge $(i,j)$ at time $t$ ($\ge 0$). We set
\begin{equation}
b_{ij}(t) = \sum_{\ell; \tilde{t}_\ell \le t} e^{-\alpha(t-\tilde{t}_\ell)}H(t-\tilde{t}_\ell),\label{eq:tie_decay_edge_update}
\end{equation}
where $H$ denotes the Heaviside step function~\cite{ahmad2021tie}. Equation~\eqref{eq:tie_decay_edge_update} indicates that $b_{ij}(t)$ increases by $1$ when an event occurs on $(i,j)$ and decays exponentially in time at rate $\alpha$ $(>0)$. We refer to this original tie-decay network model as the exponential tie-decay network. We show in Fig.~\ref{fig:spline_decay}(a) an example time course of $b_{ij}(t)$.

The exponential decay in Eq.~\eqref{eq:tie_decay_edge_update} results in a convenient updating rule for matrix $B(t)$ upon arrivals of time-stamped events at arbitrary edges and times. To elaborate, we consider a sequence of times $0\le t_1 < t_2 < \cdots < t_n$, where $t_\ell$ is the $\ell$th time at which least one event occurs between any pair of nodes in the network, and $n$ is the number of unique times at which any event occurs in the given data. Then, we obtain
\begin{equation}
B(t) = \sum_{\ell; t_\ell\le t} e^{-\alpha(t-t_\ell)}A_\ell,
\label{eq:tie_decay_matrix}
\end{equation}
where $A_\ell$ is the adjacency matrix of the input network at time $t_\ell$. In other words, the $(i, j)$ entry of $A_{\ell}$ is equal to the number of events occurring on $(i, j)$ at time $t_{\ell}$. Equation~\eqref{eq:tie_decay_matrix} yields
\begin{equation}
B(t_{\ell+1}) = e^{-\alpha(t_{\ell+1}-t_\ell)}B(t_\ell) + A_{\ell+1},
\label{eq:tie_decay_matrix_update}
\end{equation}
for each $\ell\in\{1,\ldots,n-1\}$. Moreover, for any $t\in[t_1,\infty)$, there is a unique $\ell\in\{1,\ldots,n-1\}$ such that $t_\ell \le t < t_{\ell+1}$ with the interpretation of $t_{n+1} = \infty$. Therefore,
\begin{equation}
B(t) = e^{-\alpha(t-t_\ell)}B(t_\ell),
\label{eq:tie_decay_matrix_arbitrary_update}
\end{equation}
for all $t\in[t_{\ell}, t_{\ell+1})$. By combining Eqs.~\eqref{eq:tie_decay_matrix_update} and \eqref{eq:tie_decay_matrix_arbitrary_update}, one can conveniently compute $B(t)$ at any time $t$.

\subsection{Spline tie-decay network and its mathematical properties}\label{sec:intro_spline}

As we discussed in the introduction section, there are two problems in the exponential tie-decay network. First, the edge weight $b_{ij}(t)$ discontinuously changes in time upon an event. Second, related to the first problem, the exponential tie-decay network assumes that the response of the edge weight to a single event is fast. We say so because the exponential function, $e^{-\alpha t}$ ($t\ge 0$), which represents the effect of a single event on $b_{ij}(t)$, is peaked at $t=0$. Although one can implement a delayed response by making $\alpha$ $(>0)$ small, the exponential function is still peaked at $t=0$. In contrast, to model realistic situations, we may want to assume that the effect of a single event on the edge takes some time to build up, reflecting response times needed for human individuals or the time needed for chemical or electrical signals to diffuse, for example.

A straightforward solution to these two problems is to introduce a general kernel function, $\phi(t)$ $(t\ge 0)$, which represents the contribution of the event occurring at time $0$ to the edge weight at time $t$. The exponential tie-decay network corresponds to $\phi(t) = e^{-\alpha t}$. We can resolve the two problems by imposing $\phi(0) = 0$ and using an arbitrary continuous unimodal function $\phi(t)$, for example. However, the use of a general function $\phi(t)$ disables a simple and efficient updating rule for $B(t)$ available in the case of the exponential tie-decay network. With a general $\phi(t)$, one needs to keep track of all the past event times, not just $B(t_{\ell})$ and $A_{\ell+1}$, to calculate $B(t_{\ell+1})$. The same trade-off between the analytical tractability with the use of the exponential kernel function and enhanced realism with the use of a more general $\phi(t)$ also exists for the Hawkes process~\cite{masuda2013self-exciting, masuda2020guide, kobayashi2016tideh, murayama2021modeling, nurek2023predicting}.

Therefore, as a compromise, we propose to use the cubic spline interpolation to implement $\phi(t)$ at small $t$, respecting $\phi(0) = 0$ and a delayed response nature of $\phi(t)$, and an exponential kernel function to implement $\phi(t)$ at large $t$. We stitch these two functions together by imposing that $\phi(t)$ is $C^1$-continuous in $t$. We refer to this variant of tie-decay network the spline tie-decay network.

We assume that an event occurs on edge $(i,j)$ at time $t=0$ and define $\phi(t)$ by
\begin{equation}
\phi(t) = \begin{cases}
f(t) & \text{if } t\in[0,h], \\
ke^{-\alpha(t-h)} & \text{if } t\in (h,\infty).
\end{cases}
\label{eq:spline_decay_one_kernel}
\end{equation}
In Eq.~\eqref{eq:spline_decay_one_kernel}, $f(t)$ denotes a cubic spline, and $h, k >0$ are parameters.
For continuity, we impose $f(0) = 0$.

There are variants of cubic spline depending on the choice of the boundary conditions and the number of panels. We use the clamped boundary spline interpolation with one panel. In other words, we use a single cubic polynomial to cover $t\in [0,h]$ and impose that $\phi(t)$ is $C^1$-continuous at $t=0$ and $t=h$, i.e., $f(0)=0$, which we already imposed, $\left. \frac{\text{d}f(t)}{\text{d}t} \right|_{t=0}=0$,
$f(h)= k$, and $\left. \frac{\text{d}f(t)}{\text{d}t} \right|_{t=h}= -k\alpha$.

Because $f(0)=0$ and $\left. \frac{\text{d}f(t)}{\text{d}t} \right|_{t=0}=0$, we can write
\begin{equation}
f(t) = rt^3+st^2,
\end{equation}
where $r,s\in\mathbb{R}$. Conditions $f(h)=k$ and $\left. \frac{\text{d}f(t)}{dt} \right|_{t=h}= -k\alpha$ yield
\begin{equation}
rh^3 + sh^2 = k
\label{eq:first_poly_coeff}
\end{equation}
and
\begin{equation}
3rh^2+2sh = -\alpha k,
\label{eq:second_poly_coeff}
\end{equation}
respectively. By solving Eqs.~\eqref{eq:first_poly_coeff} and \eqref{eq:second_poly_coeff}, we obtain
\begin{equation}
r = -\frac{\alpha k}{h^2} -\frac{2k}{h^3}
\label{eq:first_poly_coeff_solv}
\end{equation}
and
\begin{equation}
s = \frac{3k}{h^2} + \frac{\alpha k}{h},
\label{eq:second_poly_coeff_solv}
\end{equation}
specifying the spline kernel function. We compare in Fig.~\ref{fig:spline_decay} the time course of the edge weight between an exponential (shown in Fig.~\ref{fig:spline_decay}(a)) and spline (shown in Fig.~\ref{fig:spline_decay}(b)) tie-decay network for the same sequence of time-stamped events. The spline tie-decay network realizes a smooth time course of the edge weight, $b_{ij}(t)$, and peaks of the edge weight occurs with some delay after each event. 

We recall that $h$ and $k$ are parameters controlling the shape of $\phi(t)$. The limit $h\to 0$ corresponds to an exponential kernel as follows:
\begin{prop}
$\phi(t) = ke^{-\alpha t}$ almost everywhere as $h\to 0$.
\end{prop}
\begin{proof}
By direct calculation, we obtain
\begin{equation}
\lim_{h\to 0}\phi(t) = 
\begin{cases}
0 & \text{if } t = 0, \\
ke^{-\alpha t} & \text{if } t\in (0,\infty).
\end{cases}
\end{equation}
Therefore, $\displaystyle\lim_{h\to 0}\phi(t) = ke^{-\alpha t}$ almost everywhere on $[0,\infty)$.
\end{proof}

\begin{prop}\label{prop:spline-is-smaller}
$\phi(t) \le ke^{-\alpha (t-h)}$ for any $t\ge 0$.
\end{prop}
\begin{proof}
The statement trivially holds true with equality when $t \ge h$.

To analyze the case of $t < h$, let us consider
\begin{equation}
g_1(t) \equiv ke^{-\alpha (t-h)} - r t^3 - s t^2.
\end{equation}
We want to show $g_1(t) \ge 0$, $t \in [0, h]$.
We obtain
\begin{equation}
g'_1(t) = - \alpha k e^{-\alpha (t-h)} - 3rt^2 - 2st
= - \alpha k e^{-\alpha (t-h)} + 3\left(\frac{\alpha k}{h^2} + \frac{2k}{h^3}\right)t^2 - 2\left( \frac{3k}{h^2} + \frac{\alpha k}{h} \right)t.
\label{eq:g'_1-original}
\end{equation}
Because $g_1(h) = 0$, it suffices to show that $g'_1(t) < 0$ in $t\in [0, h)$.

By changing the variable to $v \equiv t/h$, we obtain
\begin{align}
(\text{RHS of Eq.~}\eqref{eq:g'_1-original}) &= -\alpha ke^{\alpha h(1-v)} + 3\alpha kv^2 + \frac{6k}{h} v^2 - \frac{6k}{h} v - 2\alpha kv\notag\\
&=  \alpha k\left[-e^{\alpha h(1-v)} + v^2\right] + \left(2\alpha k + \frac{6k}{h}\right) v (v-1).
\end{align}
Because $\alpha k, 2\alpha k + \frac{6k}{h} > 0$, $v \in [0, 1]$, and $g'_1(h) = 0$, it suffices to show that
\begin{equation}
g_2(v) \equiv - e^{\alpha h(1-v)} + v^2
\end{equation}
is negative in $v \in [0, 1)$. In fact, $g_2(v) < 0$ for $v \in [0, 1)$ because $g_2(1) = 0$ and
\begin{equation}
g'_2(v) = \alpha h e^{\alpha h (1-v)} + 2v > 0.
\end{equation}
Thus, we have shown that $g_1(t) \ge 0$, $t\in [0, h]$.
\end{proof}

\begin{remark}
Using Proposition~\ref{prop:spline-is-smaller}, we obtain
\begin{equation}
\left\| \phi(t) - ke^{-\alpha(t-h)}\right\|_1 = \int_{0}^h \left[ ke^{-\alpha(t-h)} - \phi(t) \right] dt 
= \frac{k(e^{\alpha h}-1)}{\alpha} - \frac{\alpha h^2 k}{12} - \frac{hk}{2} = O(h)
\end{equation}
as $h\to 0$.
\end{remark}

Now we derive a convenient updating equation for the edge weight of the spline tie-decay network.
We obtain
\begin{equation}
b_{ij}(t) = \sum_{\ell;\tilde{t}_\ell\le t}\phi(t-\tilde{t}_\ell).
\label{eq:spline_decay_weight}
\end{equation}
It should be noted that Eq.~\eqref{eq:spline_decay_weight} is reduced to Eq.~\eqref{eq:tie_decay_edge_update} when $\phi(t) = e^{-\alpha t}$. 

We show in Algorithm~\ref{alg:phylo_class} the computation of $b_{ij}(t)$ as $t$ advances from $\tilde{t}_{n}$ to $\tilde{t}_{n+1}$.
The algorithm to update $B(t)$ as $t$ advances from $t_n$ to $t_{n+1}$ is similar.
Algorithm~\ref{alg:phylo_class} works by expressing $b_{ij}(\tilde{t}_n)$ as a sum of two types of quantities. The first quantity, denoted by $c_{ij}(t)$, is the aggregate contribution of events whose effects are exponentially decaying. These events occurred before $\tilde{t}_n - h$. 
We initialize $c_{ij}(t)$ by setting $c_{ij}(0) = 0$.
The second type of quantity is the contribution of the other events, i.e., those that occurred between $\tilde{t}_n - h$ and $\tilde{t}_n$. The effect of these events is determined by the cubic polynomial part of the kernel function. For each of these events, we need to maintain $\tilde{t}_{\ell}$, which we realize by variables $m$ and $Q$ in Algorithm~\ref{alg:phylo_class}.
The algorithm updates these two types of quantity as the time advances from $\tilde{t}_n$ to $\tilde{t}_{n+1}$ as follows. First, it moves the $m$th events satisfying $\tilde{t}_n - h \le \tilde{t}_{m} \le \tilde{t}_n$ and $\tilde{t}_{m} < \tilde{t}_{n+1} - h$ from the second type to the first type. Algorithm~\ref{alg:phylo_class} does this by removing such $\tilde{t}_{m}$ from $Q$, increase $m$ by $1$, and adding the contribution of the event at $\tilde{t}_m$ to $c_{ij}(t)$. Second, it registers the ($n+1$)th event to the second type by adding $\tilde{t}_{n+1}$ to $Q$.
Algorithm~\ref{alg:phylo_class} removes the need of maintaining all the event times $\tilde{t}_1, \ldots, \tilde{t}_n$; one only needs to keep track of the event times after $t-h$, i.e., $\tilde{t}_m, \ldots, \tilde{t}_n$. The number of such recent events would not indefinitely increase because older events move from the second to the first type as time goes by.

\begin{algorithm}
\caption{Updating the edge weight of the spline tie-decay network}
\label{alg:phylo_class}
	\hspace*{\algorithmicindent} \textbf{Input} $\displaystyle c_{ij}(\tilde{t}_n)$, $m$, $Q = \{ \tilde{t}_m, \tilde{t}_{m+1}, \ldots, \tilde{t}_n\}$, where $m = \displaystyle \mathop{\mathrm{argmin}}\limits_{\ell=1,\ldots,n} \{\tilde{t}_\ell \mid \tilde{t}_n - \tilde{t}_\ell \le h\}$, and $\tilde{t}_{n+1}$
	\begin{algorithmic}[1]
	\State If $b_{ij}(t)$, $t \in (\tilde{t}_n, \tilde{t}_{n+1}]$ is requested, output $b_{ij}(t) = e^{-\alpha (t - \tilde{t}_n)} c_{ij}(\tilde{t}_n) + \displaystyle\sum_{\tilde{t}_\ell\in Q} f(t-\tilde{t}_\ell)$ 
	\State $c_{ij}(\tilde{t}_{n+1}) = e^{-\alpha(\tilde{t}_{n+1} - \tilde{t}_n)} c_{ij}(\tilde{t}_n)$
	\While{$\tilde{t}_{n+1} - \tilde{t}_m > h$}
		\State {$c_{ij}(\tilde{t}_{n+1}) \gets c_{ij}(\tilde{t}_{n+1}) + k e^{-\alpha(\tilde{t}_{n+1} - \tilde{t}_m)}$}
		\State $Q \gets Q \setminus \{\tilde{t}_m\}$ 
		\State $m \gets m+1$ 
	\EndWhile
	\State $Q \gets Q \cup \{ \tilde{t}_{n+1} \}$
	\State\Return $c_{ij}(\tilde{t}_{n+1})$, $m$, $Q$
	\end{algorithmic} 
\end{algorithm}

Algorithm~\ref{alg:phylo_class} is justified because
\begin{align}
b_{ij}(\tilde{t}_{n+1}) =& \sum_{\ell=1}^{n+1} \phi(\tilde{t}_{n+1} - \tilde{t}_{\ell}) \notag\\
=& \sum_{\ell; \tilde{t}_{n+1} - \tilde{t}_{\ell} \ge h} k e^{-\alpha(\tilde{t}_{n+1} - \tilde{t}_{\ell})} 
+ \sum_{\ell; \tilde{t}_{n+1} - \tilde{t}_{\ell} < h} f(\tilde{t}_{n+1} - \tilde{t}_{\ell}) \notag\\
=& \sum_{\ell; \tilde{t}_{n} - \tilde{t}_{\ell} \ge h} k e^{-\alpha(\tilde{t}_{n+1} - \tilde{t}_{\ell})} 
+ \sum_{\ell; \tilde{t}_{n+1} - \tilde{t}_{\ell} \ge h \text{ and } \tilde{t}_{n} - \tilde{t}_{\ell} < h} k e^{-\alpha(\tilde{t}_{n+1} - \tilde{t}_{\ell})} 
+ \sum_{\ell; \tilde{t}_{n+1} - \tilde{t}_{\ell} < h} f(\tilde{t}_{n+1} - \tilde{t}_{\ell}) \notag\\
=& e^{-\alpha(\tilde{t}_{n+1} - \tilde{t}_n)} \sum_{\ell; \tilde{t}_{n} - \tilde{t}_{\ell} \ge h} k e^{-\alpha(\tilde{t}_{n} - \tilde{t}_{\ell})} 
+ \sum_{\ell; \tilde{t}_{n+1} - \tilde{t}_{\ell} \ge h \text{ and } \tilde{t}_{n} - \tilde{t}_{\ell} < h} k e^{-\alpha(\tilde{t}_{n+1} - \tilde{t}_{\ell})} 
+ \sum_{\ell; \tilde{t}_{n+1} - \tilde{t}_{\ell} < h} f(\tilde{t}_{n+1} - \tilde{t}_{\ell}) \notag\\
=& e^{-\alpha(\tilde{t}_{n+1} - \tilde{t}_n)} c_{ij}(\tilde{t}_n)
+ \sum_{\ell; \tilde{t}_{n+1} - \tilde{t}_{\ell} \ge h \text{ and } \tilde{t}_{n} - \tilde{t}_{\ell} < h} k e^{-\alpha(\tilde{t}_{n+1} - \tilde{t}_{\ell})} 
+ \sum_{\ell; \tilde{t}_{n+1} - \tilde{t}_{\ell} < h} f(\tilde{t}_{n+1} - \tilde{t}_{\ell}).
\label{eq:update}
\end{align}
The sum of the first and second terms on the right-hand side of Eq.~\eqref{eq:update} is equal to $c_{ij}(\tilde{t}_{n+1})$.

\begin{figure}[t]
    \centering
     \includegraphics[width=1\textwidth]{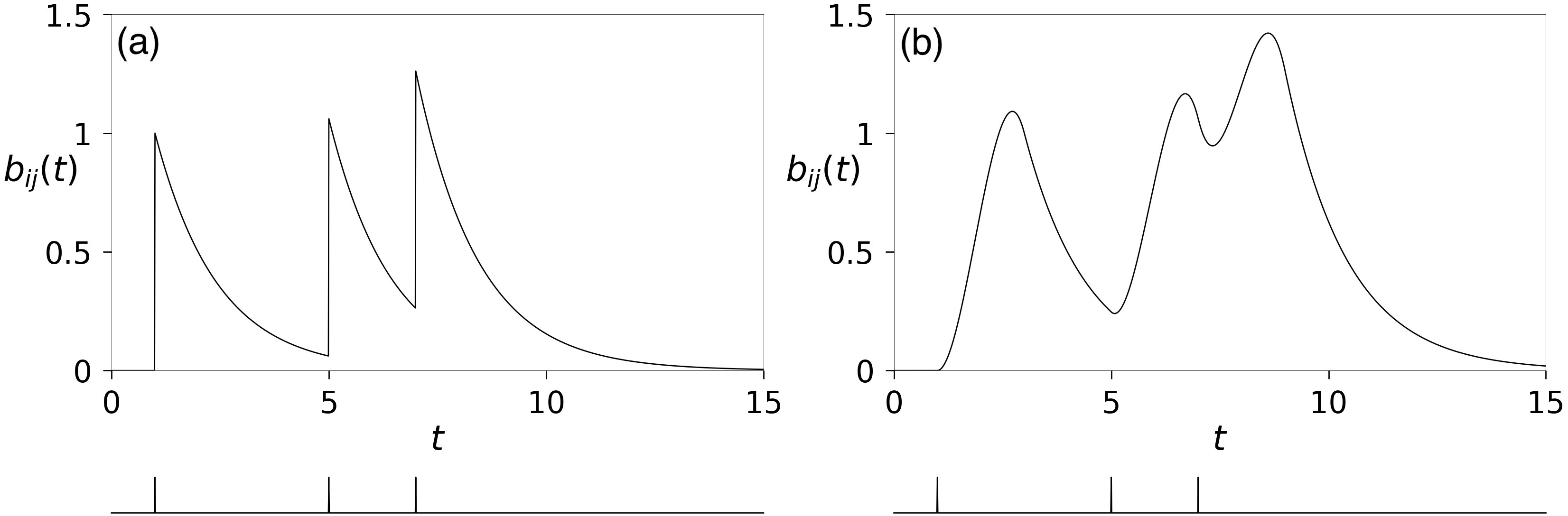}
     \caption{Time courses of the edge weight in tie-decay networks. (a) Exponential tie-decay. (b) Spline tie-decay. The vertical ticks under each plotting box represent the times of input events on the edge (i.e., $\tilde{t}_1=1, \tilde{t}_2=5$, and $\tilde{t}_3=7$). We set $\alpha=0.7$ in both (a) and (b). In (b), we set $h=2$ and $k=1$. 
     }
\label{fig:spline_decay}
\end{figure}

\subsection{A different parameterization of the spline tie-decay network}

The spline tie-decay network has three parameters, $\alpha$, $h$, and $k$.
In this section, we relate them to three parameters that intuitively characterize $\phi(t)$ in a different manner, i.e.,
the area under the curve, $\mathcal{A}$, the mean response delay time, $\mu$, and the variance of the response delay time, $\sigma^2$. Note that $\mathcal{A}$ represents the total impact of an event on the edge weight over time.

We obtain
\begin{equation}
\mathcal{A} = \int_0^{\infty} \phi(t) \text{d}t = \int_{0}^{h} (rt^3 + st^2) \text{d}t + \int_{h}^{\infty} ke^{-\alpha(t-h)} \text{d}t = \frac{rh^4}{4} + \frac{sh^3}{3}+ \frac{k}{\alpha}.
\label{eq:calc_auc}
\end{equation}
Substitution of Eqs.~\eqref{eq:first_poly_coeff_solv} and \eqref{eq:second_poly_coeff_solv} into Eq.~\eqref{eq:calc_auc} gives
\begin{equation}
\mathcal{A} = \frac{\alpha h^2k}{12} + \frac{hk}{2} + \frac{k}{\alpha}.
\label{eq:auc}
\end{equation}

Now, we impose $\mathcal{A} = 1$ as normalization, which allows us to regard $\phi(t)$ as a probability density function encoding the likelihood of the response delay time since the input event. By exploiting this interpretation, we calculate $k$, $\mu \equiv \langle \tau \rangle$, and $\langle\tau^2\rangle$ in terms of $\alpha$ and $h$,
where $\tau$ denotes the response delay time and $\langle\cdot\rangle$ represents the expectation with respect to $\phi$. First, 
Eq.~\eqref{eq:auc} combined with $\mathcal{A} = 1$ leads to
\begin{equation}
k = \frac{12\alpha}{\alpha^2 h^2 + 6\alpha h + 12}.
\label{eq:parameter_k}
\end{equation}
Second, we obtain
\begin{align}
\mu = \int_{0}^{\infty} \tau\phi(\tau) \text{d}\tau &= \int_{0}^{h} \tau (r\tau^3 + s\tau^2) \text{d}\tau + \int_{h}^{\infty} \tau ke^{-\alpha(\tau-h)} \text{d}\tau \notag\\
%
%
%
&= \frac{3\left(\alpha^3 h^3 + 7\alpha^2 h^2 + 20\alpha h + 20\right)}{5\alpha\left(\alpha^2 h^2 + 6\alpha h + 12\right)}
\label{eq:mean}
\end{align}
and
\begin{align}
\langle\tau^2\rangle = \int_{0}^{\infty} \tau^2\phi(\tau) \text{d}\tau &= \int_{0}^{h} \tau^2(r\tau^3 + s\tau^2) \text{d}\tau + \int_{h}^{\infty} \tau^2ke^{-\alpha(\tau-h)} \text{d}\tau \notag\\
%
%
%
&= \frac{2(\alpha^4 h^4 + 8\alpha^3 h^3 + 30\alpha^2 h^2 + 60\alpha h + 60)}{5\alpha^2\left(\alpha^2 h^2 + 6\alpha h + 12\right)},
\label{eq:second_order_mean}
\end{align}
where we used Eqs~\eqref{eq:first_poly_coeff_solv}, \eqref{eq:second_poly_coeff_solv}, and \eqref{eq:parameter_k}. Using Eqs.~\eqref{eq:mean} and \eqref{eq:second_order_mean}, we obtain
\begin{align}
\sigma^2 &= \langle\tau^2\rangle - \mu^2\notag\\
&= \frac{\alpha^6 h^6 + 14\alpha^5 h^5 + 99\alpha^4 h^4 + 480\alpha^3 h^3 + 1680\alpha^2 h^2 + 3600\alpha h + 3600}{25\alpha^2\left(\alpha^4 h^4 + 12\alpha^3 h^3 + 60\alpha^2 h^2 + 144\alpha h + 144\right)}.
\label{eq:variance}
\end{align}
Because we have imposed $\mathcal{A} = 1$, by setting the values of $\mu$ and $\sigma^2$, one obtains the values of $\alpha, h$, and $k$.

Next, we derive some mathematical properties of $\mu$ and $\sigma$.
\begin{theorem}
For a fixed value $\mu$, it holds true that $\sigma$ monotonically decreases as a function of $x \equiv \alpha h \in (0,\infty)$. In particular, we have $\displaystyle\lim_{x\to 0}\sigma = \mu$ and $\displaystyle\lim_{x\to\infty}\sigma=\frac{\mu}{3}$.
\label{thm:std_polynomial}
\end{theorem}
\begin{proof}
Substitution of $x=\alpha h$ in Eq.~\eqref{eq:mean} yields
\begin{equation}
\mu = \frac{3\left(x^3 + 7x^2 + 20x + 20\right)}{5\alpha\left(x^2 + 6x + 12\right)}.
\end{equation}
Therefore,
\begin{equation}
\alpha^2 = \frac{9\left(x^3 + 7x^2 + 20x + 20\right)^2}{25\mu^2\left(x^2 + 6x + 12\right)^2}.
\label{eq:alpha_sq_mean}
\end{equation}
Similarly, Eq.~\eqref{eq:variance} implies
\begin{equation}
\sigma^2 = \frac{x^6 + 14x^5 + 99x^4 + 480x^3 + 1680x^2 + 3600x + 3600}{25\alpha^2\left(x^4 + 12x^3 + 60x^2 + 144x + 144\right)},
\label{eq:var_polynomial_alpha}
\end{equation}
Substitution of Eq.~\eqref{eq:alpha_sq_mean} in Eq.~\eqref{eq:var_polynomial_alpha} yields
\begin{equation}
\sigma^2 = \frac{\mu^2(x^6 + 14x^5 + 99x^4 + 480x^3 + 1680x^2 + 3600x + 3600)}{9(x^6 + 14x^5 + 89x^4 + 320x^3 + 680x^2 + 800x + 400)}.
\label{eq:var_polynomial}
\end{equation}
Therefore, we obtain $\displaystyle\lim_{x\to 0}\sigma =  \mu$ and $\displaystyle\lim_{x\to\infty}\sigma = \frac{ \mu}{3}$.

Furthermore, Eq.~\eqref{eq:var_polynomial} yields
\begin{equation}
\frac{\text{d}\sigma^2}{\text{d}x} = \frac{-20\mu^2(x^6 + 24x^5 + 236x^4 + 1200x^3 + 3420x^2 + 5280x + 3600)}{9(x^9 + 21x^8 + 207x^7 + 1243x^6 + 4980x^5 + 13740x^4 + 26000x^3 + 32400x^2 + 24000x + 8000)}.
\label{eq:derivative_var}
\end{equation}
Because $\frac{\text{d}\sigma^2}{\text{d}x} < 0$ for all $x>0$, we conclude that $\sigma^2$ monotonically decreases as $x$ increases, and so does $\sigma$.
\end{proof}
\begin{remark}
Given a fixed value of $\sigma$, similar to Theorem~\ref{thm:std_polynomial}, one uses Eqs.~\eqref{eq:mean} and \eqref{eq:variance} to obtain
\begin{equation}
\mu = 3\sigma\sqrt{\frac{x^6 + 14x^5 + 89x^4 + 320x^3 + 680x^2 + 800x + 400}{x^6 + 14x^5 + 99x^4 + 480x^3 + 1680x^2 + 3600x + 3600}}.
\label{eq:mean_polynomial_thm}
\end{equation}
Because we have shown that the right-hand side of Eq.~\eqref{eq:var_polynomial} monotonically decreases in $x$, we find that $\mu$ monotonically increases in $x$ when $\sigma$ is fixed. Moreover, we have $\displaystyle\lim_{x\to 0}\mu = \sigma$ and $\displaystyle\lim_{x\to\infty}\mu = 3\sigma$.
\end{remark}

\section{Numerical demonstrations}

In this section, we demonstrate the spline tie-decay network in three applications: network embedding, opinion dynamics, and susceptible-infectious-recovered (SIR) epidemic dynamics.

\subsection{Data}

We use two empirical data sets of time-stamped event sequences between pairs of human individuals. A node of the network represents an individual. Both data sets are provided by SocioPatterns~\cite{stehle2011high, vanhems2013estimating, gemmetto2014mitigation}.

The first data set, referred to as Primary School, is a temporal contact network recorded from students and teachers in a primary school in Lyon, France over two days (i.e., 10/1/2009 and 10/2/2009)~\cite{stehle2011high, gemmetto2014mitigation}. The data covers the time-stamped contact events from 8:45 AM to 5:20 PM on the first day and from 8:30 AM to 5:05 PM on the second day. On each day, there were three breaks permitting students to interact with peers in other classes, which were the morning break for 20--25 minutes beginning at 10:30 AM, the lunch break between 12 PM and 2 PM, and the afternoon break for 20--25 minutes beginning at 3:30 PM.

For the embedding task, we consider all the contact events between pairs of 242 individuals; there are 8,317 edges and 125,773 time-stamped contact events on these edges. For the opinion and epidemic dynamics simulations, we only use the contact events on the first day. This is because the long intermittent period between the two days does not change the state of any node in the case of the opinion dynamics and would let all the infectious nodes recover before the second day starts in the case of the SIR dynamics. On the first day, 236 individuals appear and generate 5,901 edges and 60,623 time-stamped contact events.

The second data set, referred to as Hospital, is a temporal contact network obtained from a hospital ward in Lyon, France~\cite{vanhems2013estimating}. The network has $75$ individuals, which consist of $11$ medical doctors, $35$ nurses, and $29$ patients. The network contains $1,139$ edges and $32,424$ time-stamped contact events, which span the weekdays of a week, from 1 PM on Monday, 12/6/2010, to 2 PM on Friday, 12/10/2010. We use the network from 6 AM to 8 PM on Tuesday. The Tuesday network contains $49$ nodes with $454$ edges and $8,790$ time-stamped contact events. Note that $95.6\%$ of the contacts happen between 6 AM and 8 PM on weekdays.

\subsection{Embedding}

In this section, we investigate temporal network embedding, with which we map a tie-decay network into a trajectory in a latent space. 
The goal of embedding with the use of spline tie-decay networks is to reduce sudden changes in the embedding trajectory in $\mathbb{R}^2$ when a new event arrives. Such sudden changes occur with exponential tie-decay networks because they are discontinuous in time~\cite{thongprayoon2023embedding, thongprayoon2023online}.

\subsubsection{Landmark multidimensional scaling}

To embed the tie-decay network, we use the landmark multidimensional scaling (LMDS) \cite{Desilva2003Nips, de2004sparse, Platt2005Aistats}, which is an out-of-sample extension of the classical multidimensional scaling (MDS). The MDS and LMDS seek to find low-dimensional representations of arbitrary objects in a high-dimensional space such that the pairwise dissimilarity is preserved as much as possible. In mathematical terms, the LMDS is a map $X\mapsto\mathbb{R}^\omega$, where the dimension of $X$ is much larger than that of the embedding space, $\omega$.

We first provide a brief review on the classical MDS. Consider a set of $n$ objects $S=\{x_1,\ldots, x_n\}$ in a metric space $(X,d)$, where $d$ is a distance measure. To embed $S$ into $\mathbb{R}^\omega$ using the classical MDS, we compute the squared distance matrix $\Delta = [\Delta_{\ell\ell'}]$, where $\Delta_{\ell\ell'} = d^2(x_\ell,x_{\ell'})$. Next, we calculate the double mean-centered dot product matrix $D = -\frac{1}{2}H_n\Delta H_n$, where $H_n = I_n - \frac{1}{n}J_n$, $I_n$ is the $n\times n$ identity matrix, and $J_n$ is $n\times n$ matrix whose all entries are $1$. The embedding coordinate of $x_\ell$ is given by the $\ell$th column of the following $\omega\times n$ matrix:
\begin{equation}
L_\omega = \begin{bmatrix}
\sqrt{\lambda_1}v^{\top}_1\\
\vdots\\
\sqrt{\lambda_\omega}v^{\top}_\omega
\end{bmatrix},
\label{eq:classical_MDS}
\end{equation}
where $\lambda_\ell$ is the $\ell$th largest (positive) eigenvalue of $D$, vector $v_\ell$ is the $\ell^2$-normalized right eigenvector associated with $\lambda_\ell$, and $^\top$ represents the transposition.

The LMDS works by applying the MDS to a subset of the entire data set called landmarks. With the assistance of the landmarks, the LMDS efficiently determines the embedding coordinates of the other arbitrary points. We use $S$ as the set of landmarks and run the MDS on $S$ using Eq.~\eqref{eq:classical_MDS}. For an arbitrary point $x\in X\setminus S$, its embedding coordinate, $\psi(x)\in\mathbb{R}^\omega$, is given by
\begin{equation}
\psi(x) = -\frac{1}{2}L'_\omega\left(\delta_x - \delta_\mu\right),
\label{eq:LMDS_map}
\end{equation}
where
\begin{equation}
L'_\omega =  \begin{bmatrix}
\frac{1}{\sqrt{\lambda_1}}v^{\top}_1\\
\vdots\\
\frac{1}{\sqrt{\lambda_\omega}}v^{\top}_\omega
\end{bmatrix}\in\mathbb{R}^{\omega\times n},
\label{eq:projection_mat_LMDS}
\end{equation}
\begin{equation}
\delta_x = \begin{bmatrix}
d^2(x,x_1)\\
\vdots\\
d^2(x,x_n)
\end{bmatrix},
\label{eq:vec_sq_dist_landmark}
\end{equation}
and
\begin{equation}
\delta_\mu = \frac{1}{n}\sum_{\ell=1}^{n}\begin{bmatrix}
d^2(x_1,x_\ell)\\
\vdots\\
d^2(x_n,x_\ell)
\end{bmatrix}.
\label{eq:vec_avg_col}
\end{equation}

\subsubsection{Network embedding by landmark multidimensional scaling}\label{sec:netw_embedding}

In this section, we explain a method to embed tie-decay networks into $\mathbb{R}^\omega$ using LMDS, which we previously proposed~\cite{thongprayoon2023embedding}. The main idea is to use the tie-decay network at each time $t_\ell$ as a landmark. Recall that $t_\ell$ is the $\ell$th time at which an event occurs in the entire network. The method proceeds as follows. First, we use $\{B(t_1),\ldots, B(t_n)\}$ as the set of landmarks for the LMDS to embed $B(t)$ for $t\not\in\{t_1,\ldots, t_n\}$. Using Eq.~\eqref{eq:LMDS_map}, we obtain
\begin{equation}
\psi(B(t)) = -\frac{1}{2}L'_\omega\left(\delta_{B}(t) - \delta_\mu\right),
\label{eq:LMDS_map_network}
\end{equation}
where
\begin{equation}
\delta_{B} = \begin{bmatrix}
d^2(B(t),B(t_1))\\
\vdots\\
d^2(B(t),B(t_n))
\end{bmatrix}
\label{eq:vec_sq_dist_landmark_network}
\end{equation}
and
\begin{equation}
\delta_\mu = \frac{1}{n}\sum_{\ell=1}^{n}\begin{bmatrix}
d^2(B(t_1),B(t_\ell))\\
\vdots\\
d^2(B(t_n),B(t_\ell))
\end{bmatrix}.
\label{eq:vec_avg_col_network}
\end{equation}

We use the unnormalized Laplacian network distance, denoted by $d_{L}$, as $d$~\cite{thongprayoon2023embedding, thongprayoon2023online}. The Laplacian matrix $L(t)$ of the tie-decay network at time $t$ is defined as $L(t)\equiv \tilde{D}(t)-B(t)$, where $\tilde{D}(t)$ is the diagonal matrix whose $i$th diagonal entry is $\displaystyle\sum_{j=1}^{N}b_{ij}(t)$. We obtain
\begin{equation}
d_{L}\left(B(t_\ell), B(t_{\ell'})\right) = \sqrt{\sum_{k'=1}^{N}\left[\rho_{k'}(L(t_\ell))-\rho_{k'}(L(t_{\ell'}))\right]^2},
\label{eq:laplac_dist}
\end{equation}
where $\rho_{k'}$ is the $k'$th smallest eigenvalue of the matrix in the argument~\cite{wilson2008study, masuda2019detecting, donnat2018tracking}. One can replace $\sum_{k'=1}^{N}$ by $\sum_{k'=2}^N$ in Eq.~\eqref{eq:laplac_dist} because the smallest eigenvalue of a Laplacian matrix is always $0$.

\subsubsection{Goodness of embedding}

We evaluate the goodness of network embedding using two criteria. The first criterion stands on the fact that the MDS preserves the pairwise distances between the points in $S$ in the Euclidean space if and only if matrix $D = -\frac{1}{2}H_n\Delta H_n$ is positive semi-definite \cite{graepel1999classification, pekalska2001generalized}. Therefore, as goodness of fit, we use the ratio of the $\omega$ largest positive eigenvalues over the sum of the absolute values of all eigenvalues \cite{duin2005dissimilarity}, i.e.,
\begin{equation}
g = \frac{\displaystyle\sum_{\ell=1}^{\omega} \lambda_\ell}{\displaystyle\sum_{\ell'=1}^{n}\left|\lambda_{\ell'}\right|},
\label{eq:eigen_ratio}
\end{equation}
where $\lambda_1 >\cdots > \lambda_n$ are the eigenvalues of $D$, and we have assumed that $D$ possesses at least $\omega$ positive eigenvalues. Index $g$ ranges from $0$ to $1$. A larger $g$ value suggests better embedding.

The second criterion is the normalized stress function~\cite{borg2005modern, duin2005dissimilarity}
%
%
defined by
\begin{equation}
\theta = \sqrt{\frac{\displaystyle\sum_{\ell=1}^{n}\sum_{\ell'=1}^{\ell-1} \left[d(x_{\ell'},x_{\ell}) - \left|\psi(x_{\ell'})-\psi(x_\ell)\right|_2\right]^2}{\displaystyle\sum_{\ell=1}^{n}\sum_{\ell'=1}^{\ell-1}d^2(x_{\ell'},x_{\ell})}},
\label{eq:stress_function}
\end{equation}
where $\left|\psi(x_{\ell'})-\psi(x_\ell)\right|_2$ is the Euclidean distance in $\mathbb{R}^{\omega}$. One obtains a smaller $\theta$ if the embedding better preserves the distance between pairs of data points in the original space. A guideline of acceptable $\theta$ values is $\theta < 0.15$ \cite{borg2005modern}.

\subsubsection{Results}

We show in Figs.~\ref{fig:embedding}(a) and~\ref{fig:embedding}(b) the trajectories of Primary School network for each of the two days in the two-dimensional embedding space with the use of the spline tie-decay network. We set $\alpha = 10^{-2}$, $h=60$, and $k=1$ in this figure and Fig.~\ref{fig:hosp_embedding}. As expected, the embedded trajectories do not abruptly change as new contact events arrive. This is not the case for the trajectories in the case of the exponential tie-decay network, which we show in Figs.~\ref{fig:embedding}(c) and \ref{fig:embedding}(d) for each of the two days. We remark that Figs.~\ref{fig:embedding}(c) and \ref{fig:embedding}(d) replicate Figs.~3(a) and 3(b), respectively, in our previous article~\cite{thongprayoon2023embedding}.

In Figs.~\ref{fig:embedding}(a) and~\ref{fig:embedding}(b), the segments of the trajectory during the morning break (around the points labeled 10:30 AM) and the lunch break (around the points labeled 12 PM to 2 PM) inhabit in the second quadrant and the fourth quadrant, respectively. The segments of the trajectories during these breaks are relatively distinguishable from those for the rest of the day. This feature is shared by the trajectories with the exponential tie-decay network (see Figs.~\ref{fig:embedding}(c) and~\ref{fig:embedding}(d)).

For the present embedding with the spline tie-decay network, we obtain $g=95.23\%$ and $\theta=0.05$ with the embedding dimension $\omega=2$. These values indicate that $\omega=2$ is enough, which is the same conclusion as that for the exponential tie-decay network applied to the same data set~\cite{thongprayoon2023embedding}. As a comparison, the trajectories for the two days
with the exponential tie-decay network and $\omega = 2$ shown in Figs.~\ref{fig:embedding}(c) and~\ref{fig:embedding}(d) yield $g=95.16\%$ and $\theta=0.06$.

The trajectories of the spline tie-decay network are approximately confined in $(x,y)\in(-70,50)\times (-20,25)$. This spans a wider area than that for the exponential tie-decay network shown in Figs.~\ref{fig:embedding}(c) and~\ref{fig:embedding}(d), in which the trajectories are denser and reside approximately in $(x,y)\in (-50,40)\times (-15,20)$. Although the difference in the quality of the embedding, measured by $g$ and $\theta$, between the two cases is not substantial, we propose that trajectories that are not too much condensed in the embedding space is desirable for visualization purposes. Figure~\ref{fig:embedding} suggests that the different densities of the trajectories are not due to the difference in the scale of the embedding space, which is automatically set by the MDS.

We show in Fig.~\ref{fig:hosp_embedding}(a) the two-dimensional trajectory of the Hospital data from 6 AM to 8 PM on Tuesday using the spline tie-decay network. We show in Fig.~\ref{fig:hosp_embedding}(b) the corresponding trajectory using the exponential tie-decay network. Similar to the case of the Primary School data, the trajectory with the spline tie-decay network is continuous, smooth, and not too condensed, which contrasts with the case of the exponential tie-decay network. We obtain $g=88.41\%$ and $\theta=0.11$ for Fig.~\ref{fig:hosp_embedding}(a), and $g=87.90\%$ and $\theta=0.12$ for Fig.~\ref{fig:hosp_embedding}(b). These values assure that the embedding dimension of $\omega=2$ is adequate and that the results for the exponential kernel are consistent with our previous results~\cite{thongprayoon2023embedding}.

\begin{figure}[t]
    \centering
     \includegraphics[width=1\textwidth]{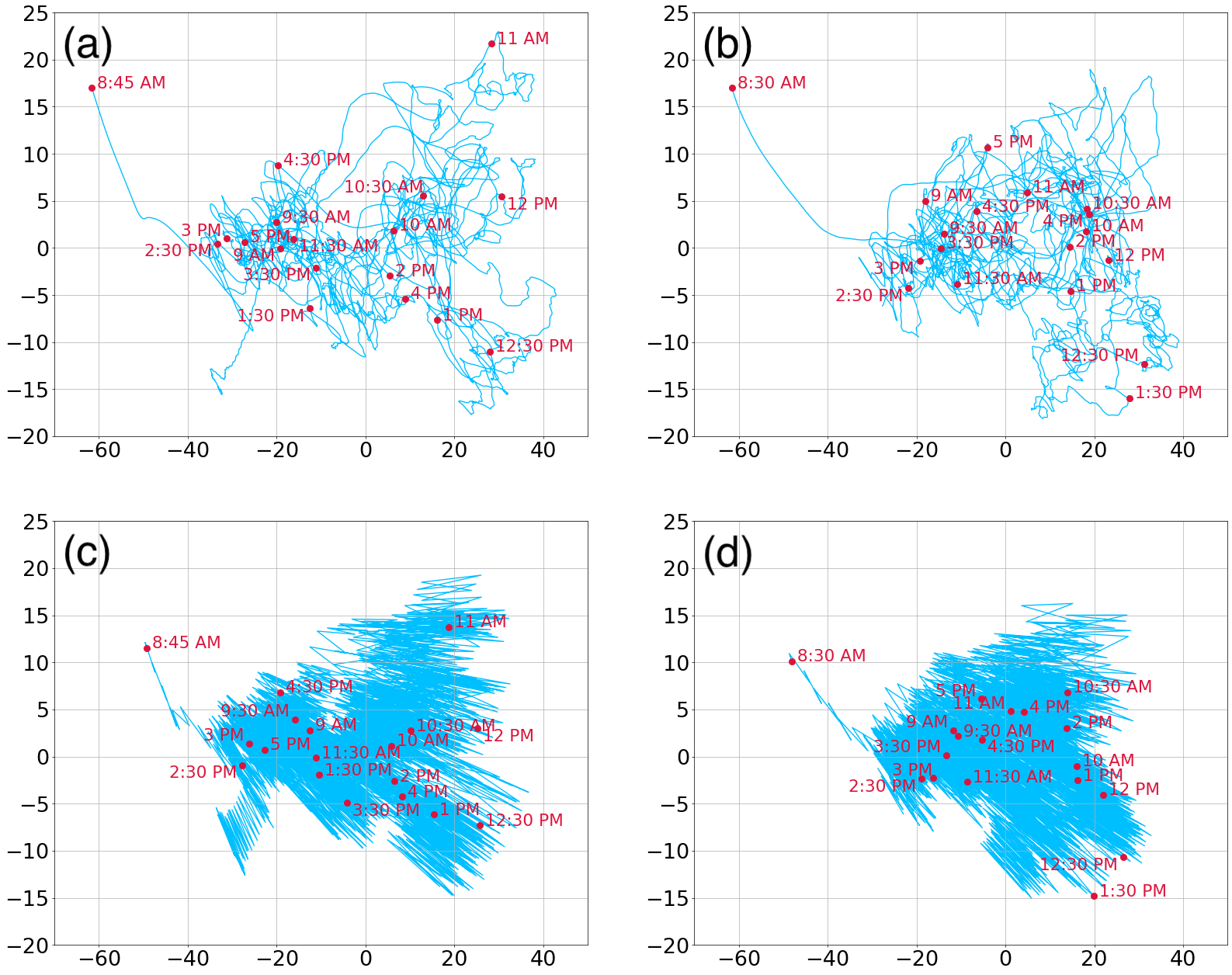}
     \caption{Trajectories of the Primary School tie-decay networks in the two-dimensional space. (a) Day 1, spline. (b) Day 2, spline. (c) Day 1, exponential. (d) Day 2, exponential. In (a) and (b), we set $\alpha = 10^{-2}$, $h=60$, and $k=1$ for the spline kernel function. In (c) and (d), we set $\alpha=10^{-2}$ for the exponential kernel function.
The initial network on each day is the null network, whose coordinate is $(x,y)\approx (-61.59, 17.00)$ in (a) and (b), and $(x,y)\approx (-51.42, 14.93)$ in (c) and (d).}
\label{fig:embedding}
\end{figure}

\begin{figure}[h]
    \centering
     \includegraphics[width=1\textwidth]{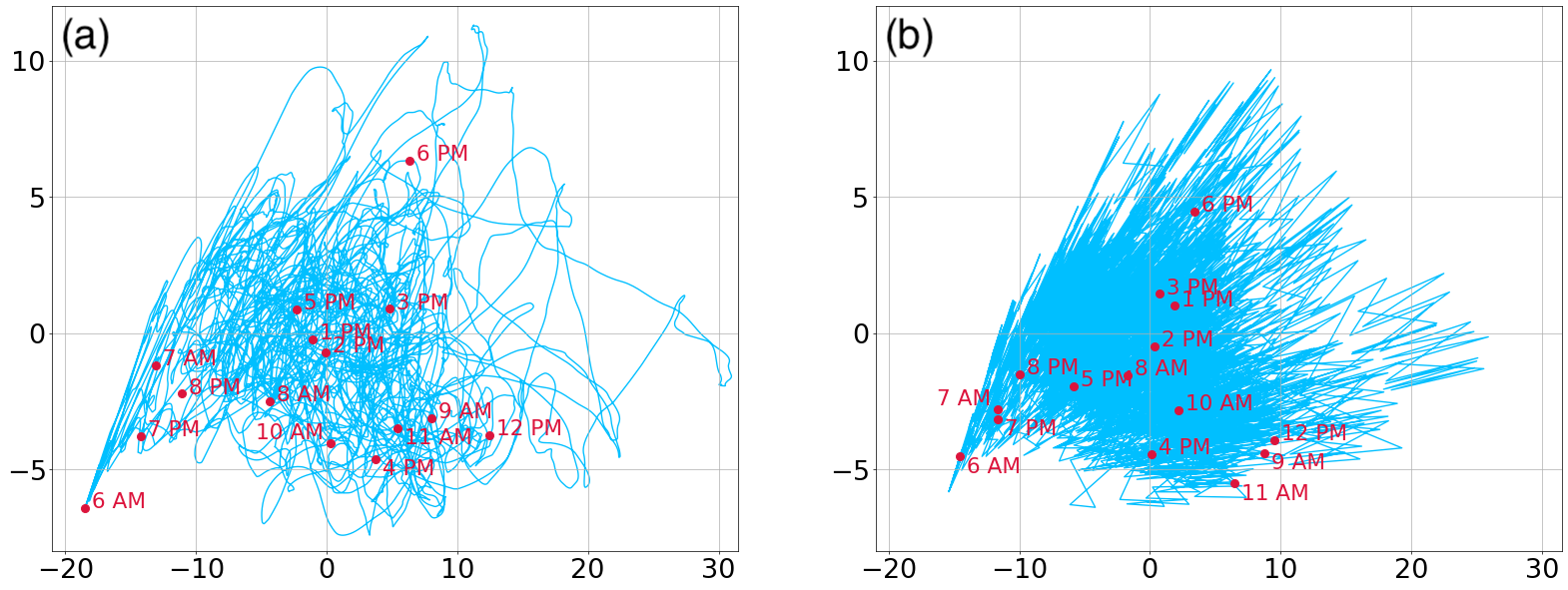}
     \caption{Trajectories of the Hospital tie-decay networks from 6 AM to 8 PM on Tuesday. (a) Spline kernel function with $\alpha = 10^{-2}$, $h=60$, and $k=1$.
      (b) Exponential kernel function with $\alpha=10^{-2}$.
The initial network is the null network, whose coordinate is $(x,y)\approx (-18.49, -6.41)$ in (a) and $(x,y)\approx (-15.39, -5.82)$ in (b).}
\label{fig:hosp_embedding}
\end{figure}

We assessed the robustness of the present embedding method by varying the values of three parameters. We show in Table~\ref{table:parameter_sensitivity} the values of $g$ and $\theta$ when one increases or decreases $\alpha$, $h$, and/or $k$ five-fold. We find that the quality of embedding in terms of $g$ and $\theta$ changes little in response to the changes in these parameter values for both Primary School and Hospital data.

\begin{table}[t]
\begin{center}
\caption{Goodness of embedding for the Primary School and Hospital tie-decay networks for various parameter values. Note that we use $(\alpha, h, k)=(10^{-2}, 60, 1)$ in the other analyses in this section. For the Primary School data, we used the events from both days for the embedding.}
\label{table:parameter_sensitivity}
\begin{tabular}{|>{\centering\arraybackslash}m{1.2cm}|>{\centering\arraybackslash}m{1.2cm}|>{\centering\arraybackslash}m{1.2cm}|>{\centering\arraybackslash}m{1.2cm}|>{\centering\arraybackslash}m{1.2cm}|>{\centering\arraybackslash}m{1.2cm}|>{\centering\arraybackslash}m{1.2cm}|}
        \hline
        \multicolumn{3}{|c|}{Parameter} & \multicolumn{2}{c|}{Primary School} & \multicolumn{2}{c|}{Hospital} \\ 
        \hline
        $\alpha$ & $h$ & $k$ & $g$ & $\theta$ & $g$ & $\theta$ \\ 
        \hline
        $10^{-2}$ & $12$ & $1$ & $95.19\%$ & $0.05$ & $88.29\%$ & $0.12$ \\ 
        \hline
        $10^{-2}$ & $12$ & $5$ & $95.19\%$ & $0.05$ & $88.29\%$ & $0.12$ \\ 
        \hline
        $10^{-2}$ & $12$ & $0.2$ & $95.19\%$ & $0.05$ & $88.29\%$ & $0.12$ \\ 
        \hline
        $10^{-2}$ & $60$ & $1$ & $95.23\%$ & $0.05$ & $88.41\%$ & $0.11$ \\
        \hline
        $10^{-2}$ & $60$ & $5$ & $95.23\%$ & $0.05$ & $88.41\%$ & $0.11$ \\
        \hline
        $10^{-2}$ & $60$ & $0.2$ & $95.23\%$ & $0.05$ & $88.41\%$ & $0.11$ \\
        \hline
        $10^{-2}$ & $300$ & $1$ & $95.41\%$ & $0.05$ & $89.20\%$ & $0.11$ \\
        \hline
        $10^{-2}$ & $300$ & $5$ & $95.41\%$ & $0.05$ & $89.20\%$ & $0.11$ \\
        \hline
        $10^{-2}$ & $300$ & $0.2$ & $95.41\%$ & $0.05$ & $89.20\%$ & $0.11$ \\
        \hline
        
        $5\times 10^{-2}$ & $12$ & $1$ & $93.81\%$ & $0.07$ & $86.75\%$ & $0.13$ \\ 
        \hline
        $5\times 10^{-2}$ & $12$ & $5$ & $93.81\%$ & $0.07$ & $86.75\%$ & $0.13$ \\ 
        \hline
        $5\times 10^{-2}$ & $12$ & $0.2$ & $93.81\%$ & $0.07$ & $86.75\%$ & $0.13$ \\ 
        \hline
        $5\times 10^{-2}$ & $60$ & $1$ & $94.76\%$ & $0.06$ & $87.77\%$ & $0.12$ \\
        \hline
        $5\times 10^{-2}$ & $60$ & $5$ & $94.76\%$ & $0.06$ & $87.77\%$ & $0.12$ \\
        \hline
        $5\times 10^{-2}$ & $60$ & $0.2$ & $94.76\%$ & $0.06$ & $87.77\%$ & $0.12$ \\
        \hline
        $5\times 10^{-2}$ & $300$ & $1$ & $95.27\%$ & $0.05$ & $89.15\%$ & $0.11$ \\
        \hline
        $5\times 10^{-2}$ & $300$ & $5$ & $95.27\%$ & $0.05$ & $89.15\%$ & $0.11$ \\
        \hline
        $5\times 10^{-2}$ & $300$ & $0.2$ & $95.27\%$ & $0.05$ & $89.15\%$ & $0.11$ \\
        \hline

	$2\times 10^{-3}$ & $12$ & $1$ & $95.76\%$ & $0.05$ & $88.44\%$ & $0.12$ \\ 
        \hline
        $2\times 10^{-3}$ & $12$ & $5$ & $95.76\%$ & $0.05$ & $88.44\%$ & $0.12$ \\ 
        \hline
        $2\times 10^{-3}$ & $12$ & $0.2$ & $95.76\%$ & $0.05$ & $88.44\%$ & $0.12$ \\ 
        \hline
        $2\times 10^{-3}$ & $60$ & $1$ & $95.79\%$ & $0.05$ & $88.50\%$ & $0.12$ \\
        \hline
        $2\times 10^{-3}$ & $60$ & $5$ & $95.79\%$ & $0.05$ & $88.50\%$ & $0.12$ \\
        \hline
        $2\times 10^{-3}$ & $60$ & $0.2$ & $95.79\%$ & $0.05$ & $88.50\%$ & $0.12$ \\
        \hline
        $2\times 10^{-3}$ & $300$ & $1$ & $95.94\%$ & $0.05$ & $88.83\%$ & $0.12$ \\
        \hline
        $2\times 10^{-3}$ & $300$ & $5$ & $95.94\%$ & $0.05$ & $88.83\%$ & $0.12$ \\
        \hline
        $2\times 10^{-3}$ & $300$ & $0.2$ & $95.94\%$ & $0.05$ & $88.83\%$ & $0.12$ \\
        \hline
    \end{tabular}
\end{center}
\end{table}

\subsection{Opinion dynamics\label{sub:opinion}}

As a second demonstration of the spline tie-decay network, we analyze a deterministic opinion dynamics model in continuous time.

\subsubsection{Model}

Let $x_i(t) \in \mathbb{R}$ be the opinion of the $i$th node at time $t$, where $i\in\{1,\ldots, N\}$.
We consider the Laplacian-driven deterministic opinion dynamics on temporal networks given by~\cite{mirzaev2013laplacian, sugishita2021opinion}
\begin{equation}
\frac{\text{d}\vec{x}}{\text{d}t} = -L(t)\vec{x}(t),
\label{eq:diff_eq_opinion}
\end{equation}
where $\vec{x}(t) = \left[x_1(t)\ \cdots\ x_N(t)\right]^{\top}$, and
we recall that $L(t)$ is the Laplacian matrix at time $t$ introduced in section~\ref{sec:netw_embedding}. We expect that the opinions of the different nodes converge to a common value. We ask the speed at which this occurs depending on the kernel of the spline tie-decay network.

\subsubsection{Simulation methods}

We integrate Eq.~\eqref{eq:diff_eq_opinion} using the Euler method with a step size of 10 seconds.
Note that the time resolution for both Primary School and Hospital data sets is 20 seconds. We confirmed that
the following results were similar when we used a single time step of 5 seconds for the Euler method. In each simulation, we initialize the opinions by $x_i(0) = 1$ for a pre-selected node $i\in\{1,\ldots, N\}$ 
and $x_j(0)=0$, $\forall j \neq i$. Every 10 seconds, we calculate the standard deviation of $\{ x_1(t), \ldots, x_N(t) \}$, which we denote by $\tilde{\sigma}(t)$. We run $N$ simulations in total, with one simulation for each $i$ to be initialized as $x_i(0) = 1$. Then, for each $t$ at which we measure $\tilde{\sigma}(t)$ (i.e., every 10 seconds), we average $\tilde{\sigma}(t)$ over all the $N$ simulations. We refer to the obtained average as the opinion spread at time $t$. 

\subsubsection{Results}

We first compare the opinion dynamics among kernel functions with different $h$ values (i.e.,
$h \in \{ 10, 400, 1000, 2000, 3000 \}$)
and the same decay rate $\alpha = 10^{-2}$. 
We show these kernel functions in Fig.~\ref{fig:opinion_spread}(a).
We show in Fig.~\ref{fig:opinion_spread}(b) the opinion spread as a function of time $t$ for the Primary School network. The figure indicates that, before $t \approx 8,620$, the opinions converge faster when $h$ is smaller. This result is intuitive because a small $h$ corresponds to a small mean response delay of the kernel function.
However, after this time, the opinions converge faster when $h$ is larger.
Figure~\ref{fig:opinion_spread}(c) shows the time courses of the opinion spread for the Hospital network. Unlike in Fig.~\ref{fig:opinion_spread}(b), the monotonic tendency between the opinion spread and $h$ is absent in early times. However, after $t \approx 30,380$, the opinions converge faster for larger $h$, which is qualitatively the same as the result for the Primary School network shown in Fig.~\ref{fig:opinion_spread}(b).
We consider that this phenomenon occurs because a kernel function with a large mean response delay (i.e., large $h$) is active (i.e., its value is larger than an arbitrary threshold value ($>0$)) for longer time than a kernel function with a short mean response delay. This factor may have driven faster convergence.

Next, we compare among kernel functions that share the same mean response delay of $150$ and have different standard deviations of the response delay. We show five such kernels in Fig.~\ref{fig:opinion_spread}(d), where $\sigma$ represents the standard deviation of each kernel. We consider $\sigma = 55$, $75$, $95$, $115$, and $135$.
We show in Figs.~\ref{fig:opinion_spread}(e) and~\ref{fig:opinion_spread}(f) time courses of the opinion spread for the Primary School and Hospital networks, respectively, under the different kernel functions shown in Fig.~\ref{fig:opinion_spread}(d). We find that the opinions converge faster in both networks when $\sigma$ is larger. Similar to the results shown in Figs.~\ref{fig:opinion_spread}(b) and~\ref{fig:opinion_spread}(c), we suggest that the faster convergence of opinions with a larger $\sigma$ may owe to a longer time window in which the kernel is active.

\begin{figure}[t]
   \centering
     \includegraphics[width=1\textwidth]{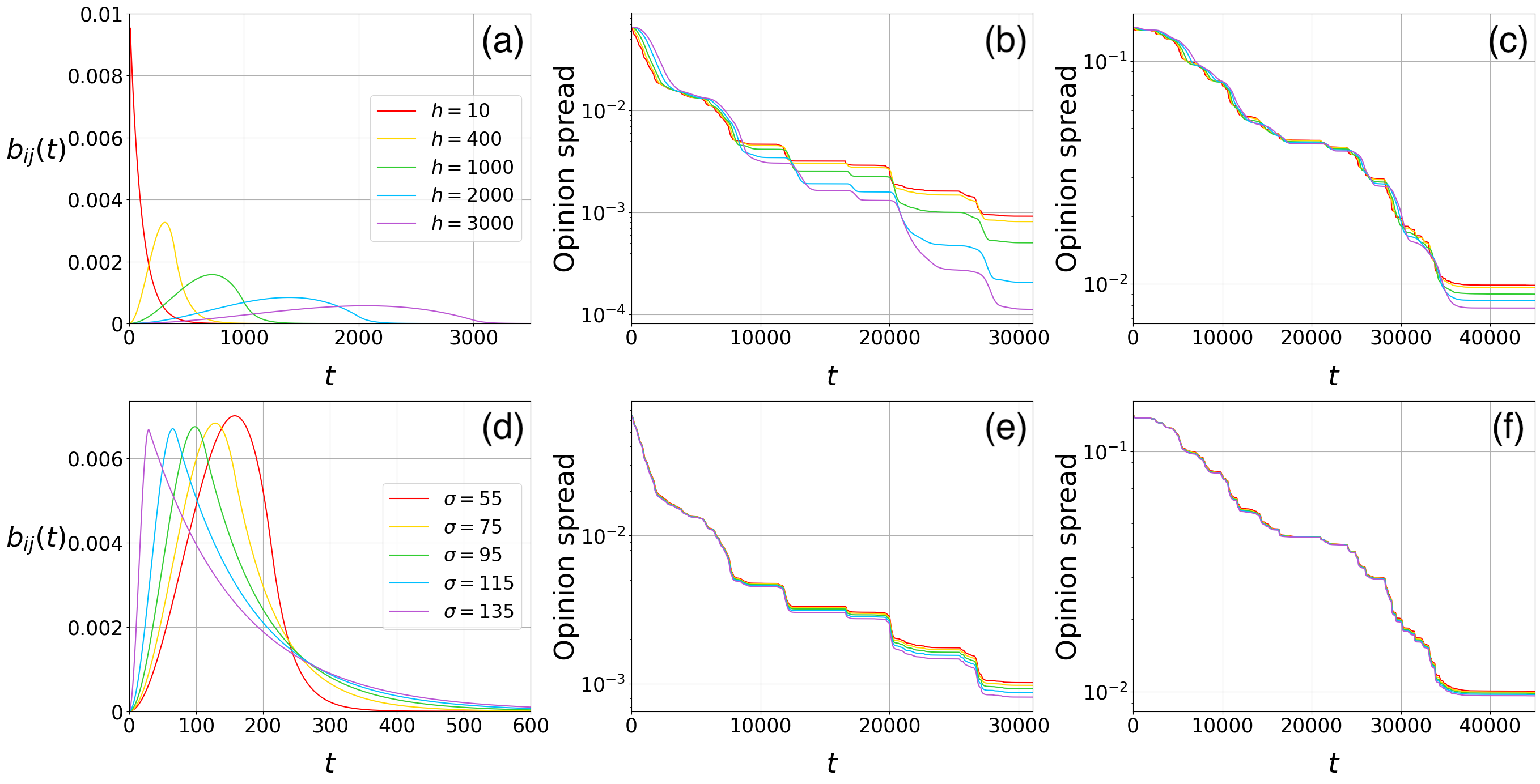}
     \caption{Opinion dynamics on spline tie-decay networks. (a) Kernels with different mean response delays. (b) Time courses of the opinion spread on the Primary School network for the kernels shown in (a). (c) Time courses of the opinion spread on the Hospital network for the kernels shown in (a). (d) Kernels with the same mean response delay of $150$ seconds and different standard deviations of the response delay, $\sigma$. (e) Time courses of the opinion spread on the Primary School network for the kernels shown in (d). (f) Time courses of the opinion spread on the Hospital network for the kernels shown in (d). In (a), we set $\alpha=10^{-2}$ for all kernels. To maintain the area under the curve $\mathcal{A} = 1$, we set $k\approx 9.52\times 10^{-3}$, $2.31\times 10^{-3}$, $6.98\times 10^{-4}$, $2.26\times 10^{-4}$, and $1.10\times 10^{-4}$ when $h=10$, $400$, $1000$, $2000$, and $3000$, respectively. In (d), for $\mathcal{A} = 1$ to be respected, the kernel with $\sigma= 55, 75, 95, 115$, and $135$ is realized by $(\alpha, h, k)\approx (0.033, 213.11, 3.82\times 10^{-3})$, $(0.015, 156.85, 5.69\times 10^{-3})$, $(0.011, 112.46, 6.27\times 10^{-3})$, $(8.78\times 10^{-3}, 70.75, 6.54\times 10^{-3})$, and $(7.42\times 10^{-3}, 30.11, 6.65\times 10^{-3})$, respectively.
}
\label{fig:opinion_spread}
\end{figure}

\subsection{SIR dynamics}

As a third demonstration of the spline tie-decay network, we investigate a stochastic epidemic dynamics model.

\subsubsection{Stochastic SIR model}

We consider the stochastic SIR model with which each node is either susceptible (S), infectious (I), or recovered (R) at any time $t$. An infectious node, denoted by $i$, independently infects each of its susceptible neighbors, denoted by $j$, at a rate $\beta b_{ij}(t) > 0$, where $\beta$ is the infection rate. Each infectious node recovers at a rate $\gamma>0$, called the recovery rate.

For both Primary School and Hospital networks, we set $\beta=0.5$. We set $\gamma=5\times 10^{-4}$ for the Primary School network and $\gamma=10^{-4}$ for the Hospital network. We use different $\gamma$ values for the two data sets to keep the epidemic dynamics ongoing until the time of later events in the respective network in many runs.

\subsubsection{Simulation methods}
\label{sec:SIR_methods}

To simulate the stochastic SIR model  on the spline tie-decay network, we employ the rejection sampling algorithm~\cite{masuda2023gillespie} with a time step of two seconds. We confirmed that the following results are quantitatively similar when we use a time step of one second.
We launch each simulation with the sole initially infectious $i$th node. All the other $N-1$ nodes are initially susceptible. For a given tie-decay network, we run the simulation $40$ times for each of the $N$ possible initially infectious nodes $i\in\{1,\ldots, N\}$. Therefore, for a given kernel function for the spline tie-decay network, we run $236 \times 40 = 9,440$ simulations for the Primary School network $(N=236)$ and $49 \times 40 = 1,960$ simulations for the Hospital network $(N=49)$. We then compute the number of recovered nodes averaged over all simulations at each time $t$, which we denote by $N_{\text{R}}(t)$.
Because a recovered node has necessarily undergone the infectious state, $N_{\text{R}}(t)$ is a standard index for quantifying the magnitude of epidemic spreading.

\subsubsection{Results}

We begin by comparing the SIR dynamics for the different kernels shown in Fig.~\ref{fig:opinion_spread}(a). Figure~\ref{fig:disease_spread}(a) shows $N_{\text{R}}(t)$ for the Primary School network. We observe larger epidemic spreading when $h$ is smaller. This result is intuitive because, if it takes longer time for the weight of the edge between an infectious node and a susceptible node to grow large after a contact event (i.e., larger $h$), then it is more probable that the infectious node recovers before infecting the susceptible neighbor.

We show $N_{\text{R}}(t)$ on the Hospital network in Fig.~\ref{fig:disease_spread}(b). Similar to Fig.~\ref{fig:disease_spread}(a), $N_{\text{R}}(t)$ roughly monotonically decreases with $h$. However, this tendency only loosely holds true for this network; the final size (i.e., $\lim_{t\to\infty} N_{\text{R}}(t)$) with $h=10^3$ (green line in Fig.~\ref{fig:disease_spread}(b)) is larger than that with $h=400$ (yellow line).

Prior research suggested that bursty nature of temporal networks can enhance epidemic spreading by having a higher frequency of short inter-event times and long inter-event times and a lower frequency of intermediately sized inter-event times compared to the Poisson process~\cite{jo2014analytically, masuda2020small}. This phenomenon is intuitively because an increased frequency of short inter-event times can enhance epidemic spreading. With spline tie-decay networks, one can compare impacts of different kernels with the same mean response delay and different standard deviation, $\sigma$, of the response delay, as we did for the opinion dynamics model in section~\ref{sub:opinion}. The kernels with a larger $\sigma$ imply that the impact of an event on the edge weight tends to be large at both short and long delay, generating a situation analogous to the case of bursty event sequences.

Motivated by this analogy, we compare epidemic spreading among the kernels shown in Fig.~\ref{fig:opinion_spread}(d). Figures~\ref{fig:disease_spread}(c) and \ref{fig:disease_spread}(d) show $N_{\text{R}}(t)$ for the different kernels for the Primary School and Hospital networks, respectively. We find that $N_{\text{R}}(t)$ is not monotonic in terms of $\sigma$ for both networks. Therefore, the results shown in Fig.~\ref{fig:disease_spread} as a whole suggest that the impact of the mean response delay (i.e., $h$) on epidemic dynamics is larger than that of the standard deviation of the response delay (i.e., $\sigma$).

\begin{figure}[t]
   \centering
     \includegraphics[width=0.66\textwidth]{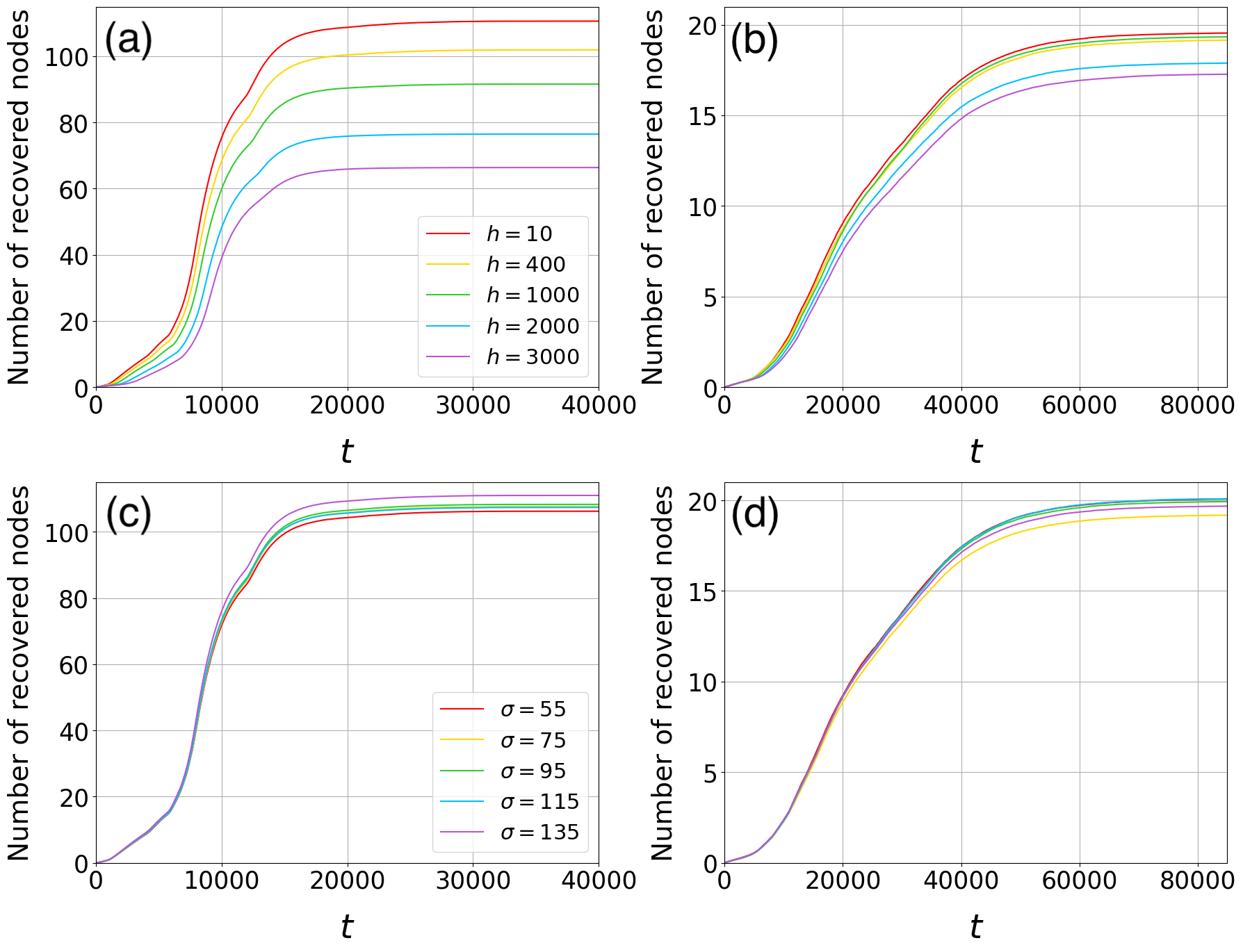}
     \caption{SIR dynamics on spline tie-decay networks. (a) Time courses of the average number of recovered nodes on the Primary School network for the kernels shown in Fig.~\ref{fig:opinion_spread}(a). (b) Time courses for the Hospital network for the kernels shown in Fig.~\ref{fig:opinion_spread}(a). (c) Time courses for the Primary School network for the kernels shown in Fig.~\ref{fig:opinion_spread}(d). (d) Time courses for the Hospital network for the kernels shown in Fig.~\ref{fig:opinion_spread}(d). In (a) and (c), we set the infection rate $\beta=0.5$ and the recovery rate $\gamma=5\times 10^{-4}$. In (b) and (d), we set $\beta=0.5$ and $\gamma=10^{-4}$.}
\label{fig:disease_spread}
\end{figure}

\section{Discussion}

We proposed the spline tie-decay network. Compared to its exponential counterpart, the spline tie-decay network enables continuity (more strongly, $C^1$-continuity) in the edge weight, which we suggest to be a practically useful property for downstream tasks such as network embedding. Our idea is to stitch a cubic polynomial and the exponential function to obtain a $C^1$-continuous kernel function. The thus obtained spline tie-decay network eliminates the need of tracking all the past event times and provides a computationally efficient algorithm to update the edge weight upon the arrival of any new event, which is a desirable feature shared with the exponential tie-decay network.

The spline kernel is the same as the exponential kernel at large $t$. One may want to use kernel functions that decay more slowly than the exponential function. In a related vein, Hawkes processes with power-law decay functions have been used for modeling seismological data \cite{vere1970stochastic, hawkes1971point, hawkes1971spectra, Ogata1988JAmStatAssoc}, video viewing activity \cite{crane2008robust}, predicting temporal patterns in a social media \cite{kobayashi2016tideh}, and modeling the spread of information \cite{murayama2021modeling}. Tie-decay networks with such a power-law decaying kernel functions can be mimicked by a mixture of exponentially decaying kernel functions, regardless of whether or not we use spline functions to guarantee the $C^1$-continuity. One can realize a mixture of exponentials by giving the weight of edge $(i, j)$ at time $t$ by $b_{ij}(t) = w_1 \tilde{b}_{1, ij}(t) + w_2 \tilde{b}_{2, ij}(t)$, where $w_1 + w_2 = 1$ and $\tilde{b}_{m, ij}(t)$ (with $m \in \{1, 2 \}$) obeys the exponential or spline tie-decay network with the exponential decay rate $\alpha_m$. Then, kernel function $\phi(t)$ is the weighted average of two kernels with different exponential decay rates. It is known that a mixture of two, or a few, exponentially decaying functions can approximate a power-law decaying function over a reasonably long scale of interest \cite{feldmann2002performance, jiang2016two-state, okada2020long-tailed}. Mixtures of infinitely many exponential functions can even create an exact power-law decaying function \cite{kurthnelson2009temporal-difference, masuda2018gillespie}. Such a mixture of exponentials combined with the tie-decay network does not disrupt fast and memory-efficient nature of the online updating algorithm for $b_{ij}(t)$. Therefore, the methods presented in this article are applicable to the cases in which $\phi(t)$ is a mixture of exponential or spline kernel functions.

In addition to the variation of the kernel function in monotonically decaying part, or at large $t$, extensions on the cubic spline polynomial part are possible. We proposed to use just one cubic polynomial covering $t\in [0, h]$, i.e., a cubic spline with one panel. Alternatively, one can introduce more panels to cover $[0, h]$. In particular, for a fixed number $q$, we first set $0 < h_1< h_2 < \cdots < h_q \equiv h$ and require $\phi(h_l)=k_l$ $(\ge 0)$ for all $l\in\{1,\ldots, q\}$. Then, one can define a cubic spline polynomial with $q$ panels supplied with preferred boundary conditions around $t=0$ and $t=h$. By choosing the $h_l$ and $k_l$ values, one can customize the shape of the kernel function. Because we anyways need to track all the event times satisfying $t - \tilde{t}_{\ell} < h$ to run the spline tie-decay network, using multiple panels for the spline part, which we have outlined here, does not essentially increase the computational or memory burden.
By the same token, other types of splines can replace the cubic spline in $[0, h]$, with a caveat that there is a general trade-off between the shape of the kernel function that can be expressed and the computational cost.

We showcased the spline tie-decay network using an opinion formation and epidemic dynamics models. Clarifying how these and other dynamics, including adaptive (i.e., co-evolutionary) network dynamics~\cite{sayama2013modeling, berner2023adaptive}, are affected by the shape of the kernel function warrants future work. Application of the spline tie-decay network to structural analysis of temporal networks, such as time-dependent centrality measures and temporal community detection~\cite{fortunato2010community, li2020exploring, masuda2020guide} may also be beneficial. Owing to the avoidance of discontinuous jumps in the edge weight, spline tie-decay networks are expected to help eliminate discontinuity in their time-dependent measurements such as centrality or communities as a function of time. This feature may be advantageous in various temporal network data analysis.

\section{Acknowledgments}

The work of Naoki Masuda was supported in part by the National Science Foundation (NSF) under grant DMS-2052720, in part by JSPS KAKENHI under grants 21H04595, 23H03414, 24K14840, and 24K03013, and in part by Japan Science and Technology Agency (JST) under grant JPMJMS2021. We also thank SocioPatterns organization for providing Primary School and Hospital data sets.

\clearpage
\bibliographystyle{unsrt}
\bibliography{bibliography.bib}

\end{document}